\documentclass[10pt]{article}
\usepackage{times}
\usepackage{algorithm} 
\usepackage{algorithmic}
\usepackage{graphicx}
\usepackage{mdwlist}
\usepackage{amssymb,amsmath}
\usepackage{color}
\usepackage{fullpage}
\usepackage{url}

\newcommand{\amolnote}[1]{\noindent{\textcolor{red}{\bf Amol note: #1}}}
\newcommand{\jiannote}[1]{\noindent{\textcolor{red}{\bf Jian note: #1}}}

\usepackage{epsfig, subfigure}
\long\def\ignore#1{}

\newtheorem{theorem}{Theorem}
\newtheorem{definition}{Definition}
\newtheorem{lemma}{Lemma}

\newtheorem{corollary}{Corollary}

\newtheorem{example}{Example}
\newenvironment{proof}{\noindent\textbf{Proof: }\ignorespaces}{}
\newcommand{\qed}{\hspace*{\fill}$\Box$\medskip}

\newcommand{\calT}{\mathcal{T}}

\newcommand{\calF}{\mathcal{F}}

\newcommand{\calC}{\mathcal{C}}
\newcommand{\Prob}{\mathsf{Pr}}
\newcommand{\calX}{\mathcal{X}}

\newcommand{\rank}{\Upsilon}

\newcommand{\eat}[1]{}
\newcommand{\Cvee}{\textcircled{\small{$\vee$}}}
\newcommand{\Cwedge}{\textcircled{\small{$\wedge$}}}
\newcommand{\Topk}{\emph{Top-$\mathrm{k}$}}
\newcommand{\UTK}{\emph{UTop-$\mathrm{k}$}}
\newcommand{\URK}{\emph{URank-$\mathrm{k}$}}

\newcommand{\PRFs}{$PRF^*$}

\newcommand{\PTK}{$PT-\rmk$}

\newcommand{\rmk}{\mathrm{k}}
\newcommand{\Exp}{\mathsf{E}}
\newcommand{\bfP}{\mathrm{P}}
\newcommand{\bfr}{\mathbf{r}}
\newcommand{\bfC}{\mathbf{C}}
\newcommand{\bfM}{\mathbf{M}}
\newcommand{\bfbr}{\bar{\mathbf{r}}}
\newcommand{\bftr}{\tilde{\mathbf{r}}}
\newcommand{\dist}{\mathsf{d}}

\newcommand{\median}{\mathtt{median}}
\newcommand{\concluster}{C{\footnotesize ONSENSUS}-C{\footnotesize LUSTERING}}

\newcommand{\rankagg}{R{\footnotesize ANK}-A{\footnotesize GGREGATION}}

\begin{document}

\title{Computing Consensus Answers in Probabilistic Databases} 
\title{Consensus Answers for Queries over Probabilistic Databases} 
\author{
Jian Li and Amol Deshpande \\[1pt]
\{lijian, amol\}@cs.umd.edu \\[1pt]
University of Maryland at College Park
}

\date{}

\maketitle

\begin{abstract}
We address the problem of finding a ``best''  deterministic query answer
to a query over a probabilistic database.
For this purpose, we propose the notion of a consensus world (or a consensus answer)
which is a deterministic world (answer) that minimizes the expected distance
to the possible worlds (answers).
This problem can be seen as a generalization of the well-studied inconsistent information aggregation
problems (e.g. rank aggregation) to probabilistic databases.
We consider this problem for various types of queries
including SPJ queries, \Topk\ queries, group-by aggregate queries, and clustering. For different
distance metrics, we obtain
polynomial time optimal or approximation algorithms for computing the consensus
answers (or prove NP-hardness).
Most of our results are for a general probabilistic database model, called
{\em and/xor tree model}, which significantly generalizes previous probabilistic database
models like x-tuples and block-independent disjoint models, and is of
independent interest.
\end{abstract}

\section{Introduction}

There is an increasing interest in uncertain and
probabilistics databases arising in application domains such as
information retrieval~\cite{dalvi:vldb04,sen:icde07}, recommendation
systems~\cite{re:icde07,re:vldb07}, mobile object data management~\cite{cheng:sigmod03}, information
extraction~\cite{gupta:vldb06}, data integration~\cite{andritsos:icde06} and
sensor networks~\cite{deshpande:vldb04}.
Supporting complex queries and decision-making on probabilistic databases is significantly
more difficult than in deterministic databases, and the key challenges include
defining proper and intuitive semantics for queries over them, and developing efficient
query processing algorithms.

The common semantics in probabilistic databases are the ``possible worlds'' semantics, where a
probabilistic database is considered to correspond to a probability distribution over a set of
deterministic databases called ``possible worlds''.
Therefore, posing queries over such a probabilistic database generates
a probability distribution over a set of
deterministic results which we call ``possible answers''.
However, a full list of possible answers together with their probabilities
is not desirable in most cases since the size of the list could be exponentially large,
and the probability associated with each single answer is extremely small.
One approach to addressing this issue is to ``combine'' the possible answers somehow to
obtain a more compact representation of the result. For select-project-join queries,
for instance, one proposed approach is to union all the possible answers, and compute
the probability of each result tuple by adding the probabilities of all possible answers
it belongs to~\cite{dalvi:vldb04}. This approach, however, can not be easily extended to other
types of queries like ranking or aggregate queries.

Furthermore, from the user or application perspective, despite the probabilistic nature of
the data, a single, deterministic query result would be desirable in most cases, on which
further analysis or decision-making could be based. For SPJ queries, this is often
achieved by ``thresholding'', i.e., returning only the result tuples with a sufficiently
high probability of being true. For aggregate queries, often expected values are returned
instead~\cite{DBLP:conf/pods/JayramMMV07}. For ranking queries, on the other hand, a
range of different approaches have been proposed to find the true ranking of the tuples. These
include \UTK , \URK ~\cite{soliman:icde07}, probabilistic threshold \Topk\ function \cite{conf/sigmod/HuaPZL08},
Global \Topk~\cite{conf/dbrank/ZhangC08}, {\em expected rank}~\cite{Cormode09}, and so on.
Although these definitions seem to reason about the ranking over probabilistic databases in some ``natural'' ways, there is a lack of a unified and systematic analysis framework
to justify
their semantics and to discriminate the usefulness of one from another.

In this paper, we consider the problem of combining
the results for all possible worlds in a systematic way
by putting it in the context of {\em inconsistent information aggregation}
which has been studied extensively in numerous contexts over the last half century.
In our context, the set of different query answers returned from possible worlds
can be thought as inconsistent information which we need to aggregate to obtain
a single representative answer.
To the best of our knowledge, this connection between query processing in probabilistic databases
and inconsistent information aggregation, though natural, has never been realized before
in any formal and mathematical way.
Concretely, we propose the notion of {\em the consensus answer}.
Roughly speaking, the consensus answer
is a answer that is {\em closest} to the answers of the possible worlds in expectation.
To measure the closeness of two answers $\tau_1$ and $\tau_2$,
we have to define suitable distance function $\dist(\tau_1,\tau_2)$
over the answer space.
For example, if an answer is a vector,
we can simply use the $L_2$ norm; whereas in other cases, for instance, \Topk\ queries,
the definition of $\dist$ is more involved.
If the most consensus answer can be taken from any point in the answer space,
we refer it as the {\em mean answer}.
A {\em median answer} is defined similarly except that the median
answer must be the answer of some possible world with non-zero probability.

From a mathematical perspective, if the distance function is properly defined to reflect the closeness of the answers,
the most consensus answer is perhaps the best deterministic
representative of the set of all possible answers
since it can be thought as the centroid of the set of points corresponding to the possible answers.
Our key results can be summarized as follows:
\begin{list}{$\bullet$}{\leftmargin 0.15in \topsep 2pt \itemsep 1pt}
    \item (Probabilistic And/Xor Tree) We propose a new model for modeling correlations,
called the {\em probabilistic and/xor tree} model, that can capture two types of correlations, mutual exclusion and coexistence.
This model generalizes the previous models such as x-tuples, and block-independent disjoint tuples model.
More important, this model admits an elegant generating functions based framework for many types of
probability computations.
    \item (Set Distance Metrics) We show that the mean and the median world can be
    found in polynomial time for the {\em symmetric difference} metric for and/xor tree model.
    For the Jaccard distance metric, we present a polynomial time algorithm to compute the mean and
    median world for tuple independent database.
    \item (\Topk\ ranking Queries)
    The problem of aggregating inconsistent rankings has been well-studied under the name of {\em rank aggregation}~\cite{conf/www/rankaggregation}.
    We develop polynomial time algorithms for computing mean and median \Topk\ answers under the
    symmetric difference metric, and the mean answers under {\em intersection metric} and {\em generalized Spearman's footrule distance}~\cite{fagin:sjdm}, for
    the and/xor tree model.
    \item (Groupby Aggregates) 
        For group by count queries, we present a 4-approximation to the problem of finding a median answer (finding
        mean answers is trivial).
    \item (Consensus Clustering) We also consider the consensus clustering problem for the and/xor tree model
    and get a constant approximation by extending a previous result~\cite{journal/jacm/ailon08}.
\end{list}

\smallskip
\smallskip
\smallskip
\noindent{\bf Outline:} We begin with a discussion of the related work (Section \ref{sec:related}).
We then define the probabilistic and/xor tree model (Section \ref{sec:model}), and present
a generating functions-based method to do probability computations on them (Section \ref{sec:and-xor}).
The bulk of our key results are presented in Sections \ref{sec:set} and \ref{sec:topk} where
we address the problem of finding consensus worlds for different set distance metrics
and for \Topk\ ranking queries respectively.
We then briefly discuss finding consensus worlds for group-by {\em count} aggregate queries and clustering
queries in Section \ref{sec:other types of queries}.

\section{Related Work}
\label{sec:related work}
\label{sec:related}
There has been much work on managing probabilistic, uncertain, incomplete, and/or fuzzy data in
database systems and this area has received renewed attention in the last few years (see e.g.
\cite{imielinski:jacm84,barbara:kde92,lakshmanan:tods97,grahne,fuhr:is97,buckles,cheng:sigmod03,dalvi:vldb04,widom:cidr05,debulletin:march2006}).
This work has spanned a range of issues from theoretical development of data models and data languages, to practical
implementation issues such as indexing techniques. In terms of representation power, most of this work has
either assumed independence between the tuples~\cite{fuhr:is97,dalvi:vldb04}, or has restricted the correlations
that can be modeled~\cite{barbara:kde92,lakshmanan:tods97,andritsos:icde06,sarma:icde06}. Several
approaches for modeling complex correlations in probabilistic databases have also been
proposed~\cite{sen:icde07,antova:sigmod07,sen:vldb08,wang:vldb08}.

For efficient query evaluation over probabilistic databases, one of the key results is the
dichotomy of conjunctive query evaluation on tuple-independent probabilistic databases by Dalvi and
Suciu~\cite{dalvi:vldb04,conf/pods/DalviS07}. Briefly the result states that the complexity of evaluating
a conjunctive query over tuple-independent probabilistic
databases is either PTIME or \#P-complete. For the former case, Dalvi and Suciu~\cite{dalvi:vldb04} also present
an algorithm to find what are called {\em safe query plans}, that permit correct {\em extensional} evaluation of the
query. Unfortunately the problem of finding consensus answers appears to be much harder;
this is because even if a query has a safe plan, the result tuples may still be arbitrarily correlated.

In recent years, there has also been much work on efficiently answering different types of queries over probabilistic databases.
Soliman et al.~\cite{soliman:icde07} first considered the problem of ranking over
probabilistic databases, and proposed two ranking functions to combine the tuple scores and probabilities.
Yi et al.~\cite{conf/icde/YiLSK08} presented improved algorithms for the same ranking functions.
Zhang and Chomicki~\cite{conf/dbrank/ZhangC08} presented a desiderata for
ranking functions and propose Global \Topk\ queries. Ming Hua et
al.~\cite{conf/icde/HuaPZL08,conf/sigmod/HuaPZL08} recently
presented a different ranking function called {\em Probabilistic threshold \Topk\ queries}.
Finally, Cormode et al.~\cite{Cormode09} also present a semantics of ranking functions and
a new ranking function called {\em expected rank}. In a recent work, we proposed a parameterized
ranking function, and presented general algorithms for evaluating them~\cite{tech}
Other types of queries have also been recently considered over probabilistic databases
(e.g. clustering~\cite{cormode:pods08},
nearest neighbors~\cite{beskales:vldb08} etc.).

The problem of aggregating inconsistent information from different sources arises in numerous disciplines
and has been studied in different contexts over decades.
Specifically, the \rankagg\ problem aims at combining $k$ different complete ranked lists $\tau_1,\ldots,\tau_k$
on the same set of objects into a single ranking, which is the best description of the combined preferences in the given lists.
This problem was considered as early as 18th century when Condorcet and Borda proposed
a voting system for elections~\cite{book/condorcet,Borda81}.
In the late 50's, Kemeny proposed the first mathematical criterion
for choosing the best ranking~\cite{Kemeny59}.
Namely, the Kemeny optimal aggregation $\tau$ is the ranking that minimizes
$\sum_{i=1}^k \dist(\tau,\tau_i)$, where $\dist(\tau_i,\tau_j)$ is the number of pairs of elements
that are ranked in different order in $\tau_i$ and $\tau_j$ (also called Kendall's tau distance).
While computing the Kemeny optimal is shown to be NP-hard~\cite{dwork_rankaggr_revisited}, 2-approximation
can be easily achieved by picking the best one from $k$ given ranking lists.
The other well-known 2-approximation is from the fact the Spearman footrule distance,
defined to be $\dist_F(\tau_i,\tau_j)=\sum_t|\tau_i(t)-\tau_j(t)|$, is within twice
the Kendall's tau distance and the footrule aggregation can be done optimally in polynomial time~\cite{conf/www/rankaggregation}.
Ailon et al.~\cite{journal/jacm/ailon08} improve the approximation ratio to $4/3$.
We refer the readers to \cite{book/hodge00} for a survey on this problem.
For aggregating \Topk\ answers, Ailon~\cite{conf/soda/ailon07} recently obtained an $3/2$-approximation
based on rounding an LP solution.

The \concluster\ problem asks for the best clustering of a set of elements which minimizes the
number of pairwise disagreements with the given $k$ clusterings. It is known to be NP-hard~\cite{Wakabayashi98}
and a 2-approximation can also be obtained by picking the best one from the given $k$ clusterings.
The best known approximation ratio is $4/3$ due to Ailon et al.~\cite{journal/jacm/ailon08}.
Recently Cormode et al.~\cite{cormode:pods08} proposed approximation algorithms for
$k$-center and $k$-median clustering problems under attribute-level uncertainty in
probabilistic databases.

\section{Preliminaries}
\label{sec:model}
We begin with reviewing the possible worlds semantics, and
introduce the probabilistic and/xor tree model.

\subsection{Possible World Semantics}

We consider probabilistic databases with both tuple-level uncertainty (the existence of a tuple is uncertain)
and attribute-level uncertainty (a tuple attribute value is uncertain). Specifically,
we denote a probabilistic relation by $R^{P}(K;A)$,
where $K$ is the {\em key} attribute, and $A$ is the {\em value} attribute\footnote{For clarity, we will assume
singleton key and value attributes.}.
For a particular tuple in $R^P$, its key attribute is certain and is sometimes called
the possible worlds key.
$R^P$ is assumed to correspond to a probability space $(PW, \Prob)$ where the set of outcomes is a
set of deterministic relations, which we call {\em possible worlds}, $PW=\{pw_1, pw_2,...., pw_N\}$.
Note that two tuples can not have the same value for the key attribute in a single possible world.
Because of the typically exponential size of $PW$, an explicit
possible worlds representation is not feasible, and hence the semantics
are usually captured implicitly by probabilistic models with polynomial size specification.

Let $T$ denote the set of tuples in all possible worlds. For ease of notation, we will use $t\in pw$ in place of
``$t$ appears in the possible world $pw$'', $\Prob(t)$ to denote $\Prob($t$\textrm{ is present})$
and $\Prob(\neg t)$ to denote  $\Prob($t$\textrm{ is not present})$.

Further, for a tuple $t^P \in R^P$, we call the certain tuples corresponding to it (with the same key value)
in the union of the possible worlds, its {\em alternatives}.


\noindent\textbf{Block-Independent Disjoint (BID) Scheme:}
BID is one of the more popular models for probabilistic databases, and assumes that different probabilistic
tuples (with different key values) are independent of each other~\cite{dalvi:vldb04,widom:cidr05,conf/pods/DalviS07,conf/dbpl/re07}.
Formally, a BID scheme has the relational schema of the from $R(K;A;\Prob)$
where $K$ is the possible worlds key, $A$ is the value attribute, and $\Prob$
captures the probability of the corresponding tuple alternative.

\subsection{Probabilistic And/Xor Tree}
We generalize the block-independent disjoint tuples model, which can capture {\em mutual exclusion}
between tuples, by adding support for {\em mutual co-existence}, and allowing these
to be specified in a hierarchical manner. Two events satisfy the mutual co-existence correlation if in any possible world, either both happen or neither occurs.
We model such correlations using a {\em probabilistic and/xor tree} (or and/xor tree for short), which also
generalizes the notions of {\em x-tuples}~\cite{sarma:icde06,conf/icde/YiLSK08},
$p$-or-sets~\cite{conf/pods/DalviS07} and tuple independent databases.
We first considered this model for tuple-level uncertainty in an earlier paper~\cite{tech},
and generalize it here to handle attribute-level uncertainty.


We use $\Cvee$ (or) to denote mutual exclusion and $\Cwedge$ (and) for coexistence.
Figure \ref{eg_possibleworld} shows two examples of probabilistic and/xor trees.
Briefly, the leaves of the tree correspond to the tuple alternatives (we abuse
the notation somewhat and use $t_i$ to denote both the tuple, and its key value).
The first tree captures a relation with four independent tuples,
$t_1, t_2, t_3, t_4$, each with two alternatives, whereas the second tree
shows how we can capture arbitrary possible worlds using an and/xor tree
(Figure \ref{eg_possibleworld}(ii) shows the possible worlds corresponding to that tree).


\eat{
\begin{minipage}{0.35\linewidth}
{
\begin{tabular}{|c|c|}
\hline
Possible Worlds & Prob \\ \hline
$pw_{1}=\{(t_3,6),(t_2,5),(t_1,1)\}$ & .3 \\
$pw_{2}=\{(t_3,9),(t_1,7),(t_4,0)\}$ & .3 \\
$pw_{3}=\{(t_3,8),(t_4,4),(t_5,3)\}$ & .4 \\
\hline
\end{tabular}
\centerline{(i)}
}
\end{minipage}
\begin{minipage}{0.70\linewidth}
{\includegraphics[height=0.9in]{andor2}}
\end{minipage}
}

\begin{figure*}[t]
    \centerline{\includegraphics[width=1.05\textwidth]{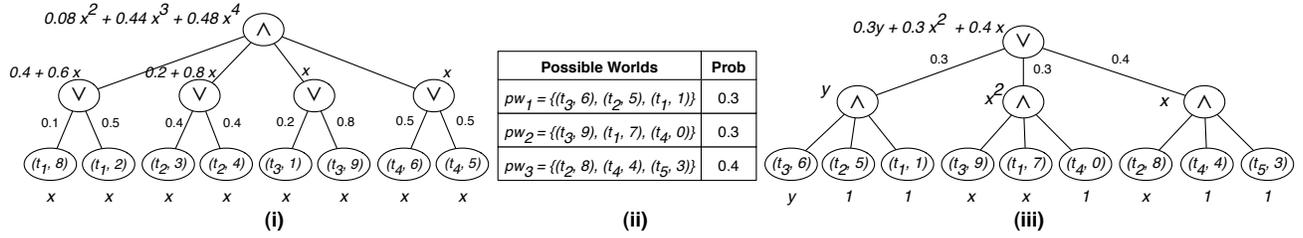}}
\caption{
(i) The and/xor tree representation of a set of block-independent disjoint tuples; the generating function
obtained by assigning the same variable $x$ to all leaves gives us the distribution over the sizes of the possible worlds.
(ii) Example of a highly correlated probabilistic database with $3$ possible worlds
and (iii) the and/xor tree that captures the correlation; the coefficient of $y$ (0.3) is $\Prob(r(t_3, 6)=1)$ (i.e., prob. that that alternative of $t_3$
is ranked at position $1$).
}
\label{eg_possibleworld}
\end{figure*}

Now, let us formally define a probabilistic and/xor tree. In tree $\calT$, we denote the set of children of node $v$ by $Ch_{\calT}(v)$ and
the least common ancestor of two leaves $l_1$ and $l_2$
by $LCA_{\calT}(l_1,l_2)$. We omit the subscript if the context is clear.

\begin{definition}
A probabilistic and/xor tree $\calT$ represents the mutual exclusion and co-existence
correlations in a probabilistic relation $R^P(K; A)$,
where $K$ is the possible worlds key, and $A$ is the value attribute.
In $\calT$, each leaf is a key-attribute pair (a tuple alternative),
and each inner node has a mark, $\Cvee$ or $\Cwedge$.
For each $\Cvee$ node $u$ and each of its children $v\in Ch(u)$,
there is a nonnegative value $\Prob(u,v)$ associated with the edge $(u,v)$.
Moreover, we require
\begin{list}{$\bullet$}{\leftmargin 0.25in \topsep 3pt \itemsep 2pt}
\item (Probability Constraint) $\sum_{v:v\in Ch(u)}\Prob(u,v)\leq 1$.
\item (Key Constraint) For any two different leaves $l_1,l_2$ holding the same key,
$LCA(l_1,l_2)$ is a $\Cvee$ node\footnote{
The key constraint is imposed to avoid two leaves with the same key coexisting in a possible world.
}.
\end{list}
Let $\calT_v$ be the subtree rooted at $v$ and $Ch(v)=\{v_1,\ldots,v_\ell\}$.
The subtree $\calT_v$ inductively defines a random subset $S_v$ of its leaves
by the following independent process:
\begin{list}{$\bullet$}{\leftmargin 0.25in \topsep 2pt \itemsep 1pt}
\item
If $v$ is a leaf, $S_v=\{v\}$.
\item
If $\calT_v$ roots at a $\Cvee$ node,
then \\ \mbox{\ \ \ \ \ \ \ } $S_v= \left\{
             \begin{array}{ll}
               S_{v_i}         & \hbox{with prob $\Prob(v,v_i)$} \\
               \emptyset       & \hbox{with prob $1-\sum_{i=1}^\ell \Prob(v,v_i)$}
             \end{array}
           \right.
$
\item
If $\calT_v$ roots at a $\Cwedge$ node,
then $S_v=  \cup_{i=1}^\ell S_{v_i}$
\end{list}
\end{definition}

\eat{
\jiannote{
Let consider a few more examples.
(1) an example for independent tuples (2)an example can't be model by BID
(3) an example for listing all possible worlds.
}
}
\noindent{Probabilistic} and/xor trees can capture more complicated correlations than the prior
models such as the BID model or x-tuples.
We remark that Markov or Bayesian network models are able to capture more general correlations~\cite{sen:icde07},
however, the structure of the model is more complex and probability
computations on them (inference) is typically exponential in the treewidth of the model.
The treewidth of an and/xor tree (viewing it as a Markov network) is not bounded, and hence
the techniques developed for those models can not be used to obtain a polynomial time algorithms
for and/xor trees.


\eat{
Figure \ref{eg_possibleworld}(i) shows an example of probabilistic
relation that models the data from a traffic monitoring
application~\cite{soliman:icde07}, where the tuples represent
automatically captured traffic data.
(ii) shows the list of all possible worlds.
We note that the tuples is correlated. For instance, tuple $t_2$ and $t_3$
are mutually exclusive since otherwise they indicate the same car being in different locations
at almost the same time.
We will introduce the probabilistic model to capture this type of correlation in following
subsection.
}

\eat{
Amol: Okay
\subsubsection{Markov Network?}
\jiannote{
maybe we don't need this.
}
}

\subsection{Computing Probabilities on And/Xor Trees}
\label{sec:and-xor}
Aside from the representational power of the and/xor tree model,
perhaps its best feature is that many types of probability computations
can be done efficiently and elegantly on them using {\em generating functions}.
In our prior work~\cite{tech}, we used a similar technique for computing
ranking functions for tuple-level uncertainty model. Here we generalize
the idea to a broader range of probability computations.

We denote the and/xor tree by $\calT$.
Suppose $\calX=\{x_1, x_2, \ldots \}$ is a set of variables. Define a mapping $s$
which associates each leaf $l \in \calT$ with a variable $s(l) \in \calX$.
Let $\calT_v$ denote the subtree rooted at $v$ and let $v_1, \ldots, v_l$ be $v$'s children.
For each node $v \in \calT$, we define a generating function $\calF_v$ recursively: 

\begin{list}{$\bullet$}{\leftmargin 0.13in \topsep 2pt \itemsep 2pt}
\item
If $v$ is a leaf, $\calF^i_v(\calX)=s(v)$.
\item
If $v$ is a $\Cvee$ node, \\[5pt]
$
\calF_v(\calX)=(1-\sum_{h=1}^l p(v,v_h))+ \sum_{h=1}^l \calF_{v_h}(\calX)\cdot p(v,v_h)
$
\item
If $v$ is a $\Cwedge$ node, $\calF^i_v(\calX)=\prod_{h=1}^l \calF_{v_h}(\calX) $.
\end{list}

\smallskip
\smallskip
The generating function $\calF(\calX)$ for tree $\calT$ is the one defined above for the root.
It is easy to see, if we have a constant number of variables, the polynomial can be expanded in
the form of $\sum_{i_1,i_2,\ldots}c_{i_1,i_2\ldots}x_1^{i_1}x_2^{i_2}\ldots$ in polynomial time.

Now recall that each possible world $pw$ contains a subset of the leaves of $\calT$ (as dictated
by the $\Cvee$ and $\Cwedge$ nodes).
The following theorem characterizes the relationship between the coefficients of $\calF$ and
the probabilities we are interested in.
\begin{theorem}
\label{thm_generating}
The coefficient of the term $\prod_{j}x_j^{i_j}$ in $\calF(\calX)$ is the total probability of
the possible worlds for which, for all $j$, there are exactly $i_j$ leaves associated with variable $x_j$.
\end{theorem}
The proof is by induction on the tree structure and is omitted.

\begin{example}
If we associate all leaves with the same variable $x$, the coefficient of $x^i$ is equal to $\Prob(|pw|=i)$.
\end{example}
The above can be used to obtain a distribution on the possible world sizes (Figure \ref{eg_possibleworld}(i)).
\begin{example}
If we associate a subset $S$ of the leaves with variable $x$, and other leaves with constant $1$,
the coefficient of $x^i$ is equal to $\Prob(|pw\cap S|=i)$.
\end{example}

\begin{example}
Next we show how to compute $Pr(r(t) = i)$ (i.e., the probability $t$ is ranked at position $i$), where $r(t)$ denote the {\em rank} of the tuple
in a possible world by some {\em score} metric. Assume  $t$ only has one alternative, $(t, a)$, and hence only
one possible value of score, $s$.
Then, in the and/xor tree $\calT$, we associate all leaves with key other than $t$
and score value larger than $s$ with variable $x$,
and the leaf $(t,a)$ with variable $y$, and the rest of leaves with constant $1$.
Then, the coefficient of $x^{j-1}y$ in the generating function is exactly
$Pr(r(t) = i)$. If the tuple has multiple alternatives, we can compute $Pr(r(t) = i)$ for it
by summing up the probabilities for the alternatives.
\end{example}
See Figure~\ref{eg_possibleworld}(iii) for an example.

\subsection{Problem Definition}

We denote the domain of answers for a query
by $\Omega$ and the distance function between two answers by $\dist()$.
Formally, we define the most consensus answer $\tau$
to be a feasible answer to the query such that the expected distance between $\tau$
and the answer $\tau_{pw}$ of the (random) world $pw$ is minimized,
i.e,
$\tau=\arg\min_{\tau'\in \Omega}\{\Exp[\dist(\tau',\tau_{pw})]\}$.

We call the most consensus answer in $\Omega$ {\em the mean answer} when $\Omega$
is the set of all feasible answers.
If $\Omega$ is restricted to be the set of possible answers (answers of some possible worlds with non-zero probability),
we call the most consensus answer in $\Omega$ {\em the median answer}.
Taking the example of the \Topk\ queries,
the median answer must be the \Topk\ answer of some possible world while
the mean answer can be any sorted list of size $\rmk$.

\section{Set Distance Measures}
\label{sec:set}
We first consider the problem of finding the consensus world for a given probabilistic database, under
two set distance measures: symmetric difference, and Jaccard distance.

\subsection{Symmetric Difference}
The symmetric difference distance between two sets $S_1$, $S_2$ is defined to be $\dist_\Delta(S_1, S_2) = |S_1 \Delta S_2| = |(S_1\setminus S_2) \cup (S_2 \setminus S_1)|$.
Note that two different alternatives of a tuple are treated as different tuples here.

\begin{theorem}
The mean world under the symmetric difference distance is the set of all tuples with probability $> 0.5$.
\end{theorem}
\begin{proof}
Suppose $S$ is a fixed set of tuples and $\bar{S}=T-S$.
Let $\delta(p)=\left\{
        \begin{array}{ll}
             1, & \hbox{if $p=true$} \\
             0, & \hbox{if $p=false$}
        \end{array}
    \right.$ be the indicator function.
We write $E_{pw\in PW}[\dist_\Delta(S,pw)]$ as follows:
{\footnotesize
\begin{align*}
&\Exp[\dist_\Delta(S,pw)] = \Exp[\sum_{t\in S} \delta(t\notin pw) + \sum_{t\in \bar{S}} \delta(t\in pw)]  \\
& = \sum_{t\in S}  \Exp[\delta(t \notin pw)] + \sum_{t\in \bar{S}}  \Exp[\delta(t \in pw)] = \sum_{t\in S} \Prob(\neg t) +\sum_{t\in \bar{S}}\Prob(t)
 \end{align*}
}

\vspace{-5pt}
Thus, each tuple $t$ contributes $\Prob(\neg t)$ to the expected distance
if $t\in S$ and $\Prob(t)$ otherwise, and hence the minimum is achieved
by the set of tuples with probability $0.5$.
\qed
\end{proof}

Finding the consensus median world is somewhat trickier, with the main concern being that the world that
contains all tuples with probability $> 0.5$ may not be a possible world.


\begin{corollary}
If the correlation can be modeled by a probabilistic and/xor tree,
the median world is the set contains all tuples with probability
greater than $0.5$.
\end{corollary}

\noindent{The} proof is by induction on the height of the tree, and is omitted for space constraints.
This however does not hold for arbitrary correlations, and it is easy to see that finding
a median world is NP-Hard even if result tuple probability computation is easy. We
show a reduction to MAX-2-SAT for a simple 2-relation query. Let the MAX-2-SAT instance consists
of $n$ literals, $x_1, \dots, x_n$, and $k$ clauses.
Consider a query $R \Join S$, where $S(x, b) = \{(x_1, 0), (x_1, 1), (x_2, 0), (x_2, 1), \dots \}$ contains two
mutually exlusive tuples each for $n$ literals; all tuples are equi-probable with probability 0.5.
$R(C, x, b)$ is a certain table, and contains two tuples for each clause: for the clause $c_1 = x_1 \vee \bar{x_2}$, it contains
tuples $(c_1, x_1, 1)$ and $(c_1, x_2, 0)$. The result of $\pi_{C} (R \Join S)$ contains one tuple for each clause, associated
with a probability of 0.75. So the median answer is the possible answer containing maximum number of tuples, which
corresponds to finding the assignment to $x_i$'s that maximizes the number of satisfied clauses.

\subsection{Jaccard Distance}
The Jaccard distance between two sets $S_1,S_2$ is defined to be
$\dist_J(S_1,S_2)={|S_1\Delta S_2|\over |S_1\cup S_2|}$.
Jaccard distance always lies in $[0,1]$ and is a real metric, i.e, satisfies triangle
inequality. Next we present polynomial time algorithms for finding the mean and median worlds
for tuple independent databases, and median world for the BID model.

\begin{lemma}
\label{lm_jaccard_gene}
Given an and/xor tree, $\calT$ and a possible world for it, $W$ (corresponding to a set of leaves of $\calT$),
we can compute $\Exp[\dist(W, pw)]$ in polynomial time.
\end{lemma}
\begin{proof}
A generating function $\calF_\calT$ is constructed with the variables associated with leaves as follows:
for $t\in W$ ($t\notin W$), the associated variable is $x$ ($y$).
For example, in a tuple independent database, the generating function is:
$$ \calF(x,y)=\prod_{t\in W}  \left(\Prob(\neg t) + \Prob(t) x\right) \prod_{t\notin W}\left( \Prob(\neg t)+ \Prob(t) y \right)$$
From Theorem \ref{thm_generating},
the coefficient $c_{i,j}$ of term $x^iy^j$ in generating function $\calF$
is equal to the total probability of the worlds such that the Jaccard distance between
those worlds and $W$ is exactly $\frac{|W|-i+j}{|W|+j}$.
Thus, the distance is $\sum_{i,j}c_{i,j}\frac{|W|-i+j}{|W|+j}$.
\end{proof}


\begin{lemma}
For tuple independent databases, if the mean world contains tuple $t_1$ but not tuple $t_2$, then $\Prob(t_1) \geq \Prob(t_2)$.
\end{lemma}
\begin{proof}
Say $W_1$ is the mean world and the lemma is not true,
i.e, $\exists t_1\in W_1, t_2 \notin W_1$ s.t. $\Prob(t_1)<\Prob(t_2)$.
Let $W=W_1-\{t_1\}$, $W_2=W+\{t_2\}$
and ${W'}=T-W-\{t_1\}-\{t_2\}$. We will prove
$W_2$ has a smaller expected Jaccard distance, thus rendering contradiction.
Suppose $|W_1|=|W_2|=k$.
We let matrix $\bfM=[m_{i,j}]_{i,j}$ where $m_{i,j}={k-i+j\over k+j}$.
We construct generating functions as we did in Lemma \ref{lm_jaccard_gene}.
Suppose $\calF_1$ and $\calF_2$ are the generating functions for $W_1$ and $W_2$, respectively.
We write $||\mathbf{A}||=\sum_{i,j}a_{i,j}$ for any matrix $\mathbf{A}$
and let $\mathbf{A}\otimes\mathbf{B}$ the Hadamard product of $\mathbf{A}$ and $\mathbf{B}$
(take product entrywise).
We denote: \\[2pt]
$\calF'(x,y)=\prod_{t\in W}  \left(\Prob(\neg t) + \Prob(t) x\right) \prod_{t\in W'}\left( \Prob(\neg t)+ \Prob(t) y \right)$ \\[2pt]
We can easily see:
$$
\calF_1(x,y)=\calF'(x,y)\left(\Prob(\neg t_1) + \Prob(t_1) x\right)\left(\Prob(\neg t_2) + \Prob(t_2) y\right)
$$
$$
\calF_2(x,y)=\calF'(x,y)\left(\Prob(\neg t_1) + \Prob(t_1) y\right)\left(\Prob(\neg t_2) + \Prob(t_2) x\right)
$$
Then, taking the difference, we get $\bar\calF=\calF_1(x,y)-\calF_2(x,y)$ is equal to:
\begin{eqnarray}
\label{eq_diff}
\calF'(x,y)\left(\Prob(\neg t_1)\Prob(t_2)-\Prob(t_1)\Prob(\neg t_2)\right)(y-x)
\end{eqnarray}
Let $\bfC_\calF=[c_{i,j}]$ be the coefficient matrix of $\calF$ where $c_{i,j}$ is the coefficient of term $x^iy^j$.
Using the proof of Lemma \ref{lm_jaccard_gene}:
\begin{eqnarray*}
\Exp[\dist(W_1,pw)]-\Exp[\dist(W_2,pw)]&=&||\bfC_{\calF_1}\otimes \bfM||-||\bfC_{\calF_2}\otimes \bfM|| \\
&=&||\bfC_{\bar\calF}\otimes \bfM ||
\end{eqnarray*}
Let $c'_{i,j}$ and $\bar{c}_{i,j}$ be the coefficient of $x^iy_j$ in $\calF'$ and $\bar\calF$, respectively.
It is not hard to see $\bar{c}_{i,j}=(c'_{i,j-1}-c'_{i-1,j})p$ from (\ref{eq_diff})
where $p=\left(\Prob(\neg t_1)\Prob(t_2)-\Prob(t_1)\Prob(\neg t_2)\right)>0$.
Then we have
\begin{eqnarray*}
||\bfC_{\bar\calF}\otimes \bfM || &=& p\sum_{i,j}\left((c'_{i,j-1}-c'_{i-1,j})m_{i,j}\right) \\
&=& p\sum_{i,j}c'_{i,j}(m_{i,j+1}-m_{i+1,j})\\
&=& p\sum_{i,j}c'_{i,j}\left({k-i+j+1\over k+j+1}-{k-i-1+j\over k+j}\right)
\end{eqnarray*}
\eat{
In fact, we can see that:

\vspace{-10pt}
{\footnotesize
\begin{align*}
&\Exp[\dist(W_1,pw)]-\Exp[\dist(W_2,pw)] & \\
&= \sum_{pw\in PW}\Prob(pw)\left(\frac{W_1\Delta pw}{W_1\cup pw} -\frac{W_2\Delta pw}{W_2\cup pw}\right)\\
&= \sum_{{t_1\in pw}\atop{t_2\notin pw}}\Prob(pw)\left({W\Delta pw-1\over W\cup pw}-{W\Delta pw+1\over W\cup pw+1}\right)\\
&\mbox{\ \ \ \ \hspace{44pt}} + \sum_{{t_1\notin pw}\atop{t_2\in pw}}\Prob(pw)\left({W\Delta pw+1\over W\cup pw+1}-{W\Delta pw-1\over W\cup pw}\right)\\
&\mbox{Rewriting\ the\ summation to be over $i = |W_1 \Delta pw|$ and $j = |W_1 \cup pw|$:} \\
&= \sum_{i,j}\left(\frac{i}{j}-\frac{i+2}{j+1}\right) \Bigg{(}\Prob\left({t_1\in pw \wedge  |W\Delta pw|=i+3\wedge \atop  t_2\notin pw\wedge |W\cup pw|=j+1}\right)\\
&\mbox{\ \ \ \ \hspace{70pt}}- \Prob\left({t_1\notin pw\wedge  |W\Delta pw|=i-1\wedge \atop t_2\in pw \wedge |W\cup pw|=j-1}\right)\Bigg{)}\\
\end{align*}
}

\vspace{-25pt}
It is easy to see $\frac{i}{j}-\frac{i+2}{j+1}<0$ for any $i\geq 0,j>0$.
In tuple independent databases, 
the last term is simply $\Prob(t_1\in pw, t_2\notin pw)-\Prob(t_2\in pw, t_1\in pw)$
which is $<0$.}
Due to the fact that ${k-i+j+1\over k+j+1}-{k-i-1+j\over k+j}>0$ for any $i,j\geq 0$,
the proof is complete.
\qed
\end{proof}

The above two lemmas can be used to efficiently find the mean world for tuple-independent databases,
by sorting the tuples in the decreasing order by probabilities, and computing the expected distance
for every prefix of the sorted order.

A similar algorithm can be used to find the median world for the BID model (by only considering the
highest probability alternative for each tuple). Finding mean worlds or median worlds under more
general correlation models remains an open problem.

\section{Top-k Queries}
\label{sec:topk}

In this section, we consider \Topk\ queries in probabilistic databases.
Each tuple $t_i$ has a score $s(t_i)$.
In the tuple-level uncertainty model, $s(t_i)$ is fixed for each $t_i$, while
in the attribute-level uncertainty model, it is an random variable.
In the and/xor tree model, we assume that the attribute field is the score
(uncertain attributes that don't contribute to the score can be ignored).
We further assume no two tuples can take the same score for avoiding ties.
We use $r(t)$ to denote the random variable indicating the rank of $t$ and
$r_{pw}(t)$ to denote the rank of
$t$ in possible world $pw$.
If $t$ does not appear in the possible world $pw$,
then $r_{pw}(t)=\infty$.
So, $\Prob(r(t)>i)$ includes the probability that $t$'s rank
is larger than $i$ and that $t$ doesn't exist.
We say $t_1$ {\em ranks higher} than $t_2$ in possible world $pw$ if
$r_{pw}(t_1) < r_{pw}(t_2)$.

Finally, we use the symbol $\tau$ to denote
rankings, and $\tau^i$ to denote the restriction of the
$\Topk$ list $\tau$ to the first $i$ items.
We use $\tau(i)$ to denote the $i^{th}$ item in the list $\tau$ for positive integer $i$,
and $\tau(t)$ to denote the position of $t\in T$ in $\tau$.




\subsection{Distance between Two $\Topk$ Answers}
\eat{
\jiannote{
state somewhere the PRF function is really preferred.
The main concern is the actual quality and running time of previous algorithms for kendall distance.
}
}
Fagin et al.~\cite{fagin:sjdm} provide a comprehensive analysis of the problem of comparing
two \Topk\ lists. 
They present extensions of the Kendall's tau and
Spearman footrule metrics (defined on full rankings) to \Topk\ lists and propose
several other natural metrics, such as the intersection metric
and Goodman and Kruskal's gamma function.
In our paper, we consider three of the metrics discussed in that paper:
the symmetric difference metric, the intersection metric and
one particular extension to Spearman's footrule distance.
We briefly recall some definitions here. For more details and the relation between
different definitions, please refer to \cite{fagin:sjdm}.

Given two $\Topk$ lists, $\tau_1$ and $\tau_2$, the normalized symmetric difference
metric is defined as: \\[4pt]
\centerline{$
\dist_{\Delta}(\tau_1,\tau_2)={1\over 2\rmk}|\tau_1\Delta
\tau_2|={1\over 2\rmk}|(\tau_1\backslash \tau_2)\cup (\tau_2\backslash \tau_1)|.
$}

\vspace{2pt}
While $\dist_{\Delta}$ focuses only on the membership,
the intersection metric $\dist_I$ also takes the order of tuples into consideration.
It is defined to be:\\[4pt]
\centerline{$\dist_I(\tau_1,\tau_2)={1\over \rmk} \sum_{i=1}^{\rmk}
\dist_{\Delta}(\tau^i_1,\tau^i_2)
$}\\[4pt]
Both $\dist_{\Delta}$ and $\dist_I()$ values are always between $0$
and $1$.

\smallskip
The original Spearman's Footrule metric is defined as the $L_1$ distance between two permutations
$\sigma_1$ and $\sigma_2$. Formally,
$F(\sigma_1,\sigma_2)=\sum_{t\in T} |\sigma_1(t)-\sigma_2(t)|$.
Let $\ell$ be a integer greater than $\rmk$.
The {\em footrule distance with location parameter $\ell$},
denoted $F^{(\ell)}$ generalizes the original footrule metric.
It is obtained by placing all missing elements in each list at position $\ell$ and then computing
the usual footrule distance between them.
A natural choice of $\ell$ is $k+1$ and we denote $F^{(\ell+1)}$ by $\dist_F$.
It is also proven that $\dist_F$ is a real metric and a member of a big and important equivalence class
\footnote{
All distance functions in one equivalence class are bounded by each other within a constant factor.
This class includes several extensions of Spearman's footrule and Kendall's tau metrics.
}
\cite{fagin:sjdm}.

It is shown in \cite{fagin:sjdm} that:
\begin{eqnarray*}
&&\dist_F(\tau_1,\tau_2)=(\rmk+1)|\tau_1\Delta\tau_2|\\
&&+\sum_{t\in \tau_1\cap\tau_2}|\tau_1(t)-\tau_2(t)|
-\sum_{t\in \tau_1\setminus\tau_2} \tau_1(t) -\sum_{t\in \tau_2\setminus\tau_1}\tau_2(t).
\end{eqnarray*}


\label{sec_expgood}

\eat{All different definitions of $\Topk$ queries in probabilistic
database proposed are attempting to trade off tuples' existence
probability and score and find those with both higher probability
can score.
However, all prior works use some ``natural'' definitions
and it is hard to tell which one is better.
We pose the problem of finding the ``best'' $\Topk$ tuples in
probabilistic database as the following optimization problem.
Suppose $d()$ is the distance between two $\Topk$ lists. Let
$\tau_{pw}$ be the $\Topk$ answer of possible world $pw$. Our goal
is to determine $\Topk$ answer $\tau$ such that
the expected distance between $\tau$ and the actual (random) $\Topk$ answer $\tau_{pw}$, i.e,
$$
\Exp[d(\tau,\tau_{pw})]
$$
is minimized.}

Next we consider the problem of evaluating consensus answers for these distance metrics.

\subsection{Symmetric Difference and \PTK\ function}
\eat{
\PTK\ function is a special class of \PRFs\ function such that
the weight function $f(i)$ takes nonzero value only if $i\leq\rmk$.
First, we address the relation between \PTK\ function and the expected distance
minimization problem.
Basically, we show the $\Topk$ tuples returned by \PTK\ function
is a minimizer of a certain class of distance functions. In particular,
the Probabilistic Threshold $\Topk$ function minimizes
the expected normalized asymmetric difference $d_{\Delta}$.
}
In this section, we show how to find mean and median \Topk\ answers under symmetric difference metric in the
and/xor tree model.
The probabilistic threshold \Topk\ (\PTK) query~\cite{conf/sigmod/HuaPZL08} has been proposed
for evaluating ranking queries over probabilistic databases, and essentially returns
all tuples $t$ for which $\Prob(r(t)\leq \rmk)$ is greater than a given threshold.
If we set the threshold carefully so that the \PTK\ query returns $\rmk$ tuples,
we can show that the answer returned is the mean answer under symmetric difference metric.

\begin{theorem}
\label{thm_mindis_prfk}
If $\tau=\{\tau(1),\tau(2),\ldots,\tau(\rmk)\}$ is the set of
$\rmk$ tuples with the largest $\Prob(r(t)\leq \rmk)$,
then $\tau$ is the mean \Topk\ answer under metric $\dist_{\Delta}$, i.e., the answer minimizes $\Exp[\dist_{\Delta}(\tau,\tau_{pw})]$.
\end{theorem}
\begin{proof}
Suppose $\tau$ is fixed.
We write $\Exp[\dist_{\Delta}(\tau,\tau_{pw})]$ as follows:

\vspace{-15pt}
{\footnotesize
\begin{align*}
\Exp[\dist_{\Delta}(\tau,\tau_{pw})]&= \Exp[\sum_{t\in T} \delta(t\in \tau \wedge t\notin \tau_{pw})+\delta(t\in \tau_{pw} \wedge t\notin \tau)]  \\
&= \sum_{t\in T\setminus \tau} \Exp[\delta(t\in \tau_{pw})] +\sum_{t\in \tau}  \Exp[\delta(t\notin \tau_{pw})] \\
&= \sum_{t\in T\setminus\tau} \Prob(r(t)\leq \rmk) +\sum_{t\in \tau} \Prob(r(t)> \rmk) \\
&= \rmk + \sum_{t\in T} \Prob(r(t)\leq\rmk) - 2\sum_{t\in \tau} \Prob(r(t)\leq \rmk)
\end{align*}
}

\vspace{-10pt}
The first two terms are invariant with respect to $\tau$.
Therefore, it is clear that the set of $\rmk$ tuples with the largest $\Prob(r(t)\leq \rmk)$
minimizes the expectation.
\qed
\end{proof}

\noindent{T}o find a median answer, we essentially need to find
the \Topk\ answer $\tau$ of some possible world such that $\sum_{t\in \tau} \Prob(r(t)\leq \rmk)$ is maximum.
Next we show how to do this given an and/xor tree in polynomial time.

We write $P(t)=\Prob(r(t)\leq \rmk)$ for ease of notation.
We use dynamic programming over the tree structure.
For each possible attribute value $a\in A$,
let $\calT^a$ be the tree which contains all leaves with attribute value at least $a$.
We recursively compute the set of tuples $pw^a_{v,i}$, which maximizes the value $\sum_{t\in pw^a_{v,i}}P(t)$ among all
possible worlds generated by the subtree $\calT^a_v$ rooted at $v$ and of size $i$,
for each node $v$ in $\calT^a$ and $1\leq i\leq \rmk$.
We compute this for all different $a$ values, and
the optimal solution can be chosen to be $\min_a(pw^a_{r,\rmk})$.

Suppose $v_1,v_2,\ldots,v_l$ are $v$'s children.
The recursion formula is:
\begin{list}{$\bullet$}{\leftmargin 0.15in \topsep 2pt \itemsep 1pt}
\item
If $v$ is a $\Cvee$ node,
$pw^a_{v,i}=\arg\max_{pw\in PW(\calT^a_{v_i})}\sum_{t\in pw}P(t)$.
\item
If $v$ is a $\Cwedge$ node,
$pw^a_{v,i}=\cup_j pw_j$ such that $\sum_j |pw_j|=i, pw_j\in PW(\calT^a_{v_j})$ and
$\sum_{t\in \cup_j pw_j}P(t)$ is maximized.
\end{list}
In the latter case, the maximum value can be computed by dynamic programming again as follows.
Let $pw^a_{[v_1,\ldots v_h],i}=\cup_{j=1}^h pw_j$ such that $\sum_{j=1}^h |pw_j|=i,pw_j\in PW(\calT^a_{v_j})$ and
$\sum_{t\in \cup_{j=1}^h pw_j}P(t)$ is maximized.
It can be computed recursive by seeing
$pw^a_{[v_1,\ldots v_h],i}=pw^a_{[v_1,\ldots v_{h-1}],p}\cup pw^a_{v_h,q}$ for $p,q$ such that
$p+q=i$ and
$\sum_{t\in pw^a_{[v_1,\ldots v_{h-1}],p}\cup pw^a_{v_h,q}}P(t)$ is maximized.
Then, it is easy to see $pw^a(v,i)$ is simply $pw^a([v_1,\ldots,v_l],i)$.

\begin{theorem}
The median $\Topk$ answer under symmetric difference metric
can be found in polynomial time for a probabilistic and/xor tree.
\end{theorem}


\subsection{Intersection Metric}

Note that the intersection metric $\dist_I$ is a linear combination of the
normalized asymmetric difference metric $\dist_\Delta$.
Using a similar approach used in the proof of Theorem \ref{thm_mindis_prfk},
we can show that:
\begin{eqnarray*}
    \Exp[\dist_I(\tau,\tau_{pw})] = {1\over \rmk} \sum_{i=1}^{\rmk} \Exp[\dist_{\Delta}(\tau^i,\tau^i_{pw})]  \mbox{\ \hspace{1.2in}}\\
= {1\over \rmk} \sum_{i=1}^{\rmk} {1\over i}\left( \rmk + \sum_{t\in T} \Prob(r(t)\leq\rmk) - 2\sum_{t\in \tau^i} \Prob(r(t)\leq i)\right) \\
\end{eqnarray*}

Thus we need to find $\tau$ which maximizes the last term,
$A(\tau)=\sum_{i=1}^{\rmk} \left({1\over i} \sum_{t\in \tau^i} \Prob(r(t)\leq i)\right)$.
We first rewrite the objective as follows, using the indicator ($\delta$) function:
\begin{eqnarray*}
A(\tau)&=&\sum_{i=1}^{\rmk} \left({1\over i} \sum_{t\in T} \Prob(r(t)\leq i))\delta(t\in \tau^i)\right) \\
&=& \sum_{t\in T}\left( \sum_{i=1}^{\rmk} {1\over i} \Prob(r(t)\leq i)\sum_{j=1}^i\delta(t=\tau(j)) \right) \\
&=&\sum_{t\in T}\sum_{j=1}^{\rmk} \left(\delta(t=\tau(j)) \sum_{i=j}^{\rmk} {1\over i}\Prob(r(t)\leq i) \right)
\end{eqnarray*}
The last equality holds since $\sum_{i=1}^{\rmk}\sum_{j=1}^i a_{ij}=\sum_{j=1}^{\rmk}\sum_{i=j}^k a_{ij}$.

\eat{
The problem can be solved in polynomial time
by reducing it to maximum weight perfect matching problem as follows.
create a bipartite graph $G(U;V,E)$ with $U=\{u_1,u_2,\ldots,u_n\}$
and $V=\{v_1,v_2,\ldots,v_{\rmk},w_1,w_2,\ldots,w_{n-\rmk}\}$.
Basically, node $u_t$ corresponds to tuple $t$ in $T$, $v_i$ corresponds to
$\tau(i)$, i.e, the $i$th tuple in the $\Topk$ answer $\tau$, and $w_i$s are dummy nodes which
ensure the existence of a perfect matching.
Define the weight of edge $(u_t,v_j)$ to be
the contribution of tuple $t$ to the objective $A(\tau)$ if $t=\tau(j)$, i.e,
$\sum_{i=j}^{\rmk}\left({1\over i}\Prob(r(t)\leq i)\right)$.
All other edges have weight $0$.
Then, we compute a maximum weight perfect matching $M$ on $G$.
Edge $(u_t,v_i)$ in $M$ indicate that $\tau(i)=t$, i.e, $t$ will be the Top-$i$ tuple.

\amolnote{I believe the above is called an {\em assignment problem}, with tuples acting as agents and each of the
top-k spots as the tasks (\url{http://en.wikipedia.org/wiki/Assignment_problem})}.
}
\smallskip
The optimization task can thus be written as an {\em assignment problem},
with each tuple $t$ acting as an agent and each of the \Topk\ positions $j$
as a task. Assigning task $j$ to agent $t$ gains a profit of $\sum_{i=j}^{\rmk} {1\over i}\Prob(r(t)\leq i)$ and the
goal is to find an assignment such that each task is assigned to at most one agent, and the profit is maximized.
The best known algorithm for computing the optimal assignment runs in $O(n\rmk \sqrt n)$ time, via computing a maximum
weight matching on bipartite graph~\cite{conf/focs/matching80}.

\smallskip
\noindent{\bf Approximating the Intersection Metric:}
We define the following ranking function, where $H_k$ denotes the $k^{th}$ Harmonic number:
$$\rank_H(t)=\sum_{i=1}^{\rmk}(H_{\rmk}-H_{i-1})\Prob(r(t)=i)=\sum_{i=1}^{\rmk}{\Prob(r(t)\leq i)\over i}.$$
This is a special case of the parameterized ranking function proposed
in~\cite{tech} and can be computed in $O(n\rmk\log^2 n)$ time for all tuples in the and/xor tree.
We claim that the \Topk\ answer $\tau_H$ returned by $\rank_H$ function, i.e., the $\rmk$ tuples with the highest $\rank_H$ values,
is a good approximation of the mean answer with respect to the intersection metric
by arguing that $\tau_H=\{t_1,t_2,\ldots,t_{\rmk}\}$ is actually an approximated maximizer of $A(\tau)$.
Indeed, we prove the fact that $A(\tau_H)\geq {1\over H_\rmk}A(\tau^*)$ where
$\tau^*$ is the optimal mean \Topk\ answer.

\begin{figure*}[t]
    {\footnotesize
\begin{eqnarray*}
    \Exp[F^*(\tau,\tau_{pw})] &=& \Exp
    \left[
        (\rmk+1)|\tau\Delta\tau_{pw}|+\sum_{t\in \tau\cap\tau_{pw}}|\tau(t)-\tau_{pw}(t)|
        -\sum_{t\in \tau\setminus\tau_{pw}} \tau(t) -\sum_{t\in \tau_{pw}\setminus\tau}\tau_{pw}(t)
    \right] \\
    &=& (\rmk+1)\Exp[|\tau\Delta\tau_{pw}|]
        +\sum_{t\in T}\Exp\left[\delta(t\in \tau\cap\tau_{pw}) |\tau(t)-\tau_{pw}(t)|\right]
      -\sum_{t\in T}\Exp\left[\delta(t\in \tau\setminus\tau_{pw})\tau(t)\right]
        -\Exp\left[\sum_{t\in \tau_{pw}\setminus\tau}\tau_{pw}(t)\right] \\
    &=& (\rmk+1)\Exp[|\tau\Delta\tau_{pw}|]
        +\sum_{t\in T}\sum_{i=1}^{\rmk}\sum_{j=1}^{\rmk}\Exp\left[
            \delta(t\in \tau\cap\tau_{pw})\delta(t=\tau_{pw}(i))\delta(t=\tau(j))|i-j|
            \right] \\
    & & -\sum_{t\in T}\sum_{i=1}^{\rmk}\Exp\left[ \delta(t\in \tau\setminus\tau_{pw})\delta(t=\tau(i))i\right]
        -\sum_{t\in T\setminus\tau}\rank_2(t)\\
    &=& (\rmk+1) \Exp[|\tau\Delta\tau_{pw}|]
        +\sum_{t\in T}\sum_{i=1}^{\rmk}\left(\delta(t=\tau(i))\sum_{j=1}^{\rmk}\Prob(r(t)=j)|i-j|\right)
     -\sum_{t\in T}\sum_{i=1}^{\rmk}\left(\delta(t=\tau(i))i\Prob(r(t)>\rmk)\right)
        -\sum_{t\in T\setminus\tau}\rank_2(t) \\
    &=& (\rmk+1)(\rmk+\sum_{t\in T}\rank_1(t)-2\sum_{t\in \tau}\rank_1(t))
        +\sum_{t\in T}\sum_{i=1}^{\rmk}\delta(t=\tau(i))\rank_3(t,i)-\sum_{t\in T\setminus\tau}\rank_2(t)\\
    &=& (\rmk+1)\rmk+ \sum_{t\in T}\left((\rmk+1)\rank_1(t)-\rank_2(t)\right)
        +\sum_{t\in T}\sum_{i=1}^{\rmk}\delta(t=\tau(i))(\rank_3(t,i)+\rank_2(t)-2(\rmk+1)\rank_1(t))
\end{eqnarray*}
}
\vspace{-10pt}
\caption{Derivation for Spearman's Footrule Distance}
\vspace{-5pt}
\label{fig:spearmans}
\end{figure*}

Let $B(\tau)=\sum_{t\in \tau}\rank_H(t)$ for any \Topk\ answer $\tau$.
It is easy to see $A(\tau^*)\leq B(\tau^*)\leq B(\tau_H)$ since $\tau_H$ maximizes the $B()$ function.
Then, we can get:
\begin{eqnarray*}
\label{eq_AB}
A(\tau_H)&=&    \sum_{j=1}^{\rmk} \sum_{i=j}^{\rmk} {1\over i}\Prob(r(t_j)\leq i) \\
         &\geq& \sum_{j=1}^{\rmk} ({H_\rmk-H_{j-1}\over H_k})\sum_{i=1}^{\rmk} {1\over i}\Prob(r(t_j)\leq i) \\
         &=& \sum_{j=1}^{\rmk} ({H_\rmk-H_{j-1}\over H_k})\rank_H(t_j) \\
         &\geq& {1\over \rmk}\sum_{i=1}^\rmk ({H_\rmk-H_{i-1}\over H_k}) \sum_{i=1}^\rmk \rank_H(t_i) \\
         &=& {1\over H_\rmk} B(\tau_H) \geq {1\over H_\rmk}A(\tau^*).
\end{eqnarray*}
The second inequality holds because for non-decreasing sequences $a_i(1\leq i\leq n)$ and $c_i(1\leq i\leq n)$, \\
\centerline{$\sum_{i=1}^n a_ic_i\geq {1\over n}(\sum_{i=1}^n a_i)(\sum_{i=1}^n c_i)$}
\subsection{Spearman's Footrule}

\noindent{For} a \Topk\ answer $\tau=\{\tau(1),\tau(2),\ldots,\tau(\rmk)\}$, we define:
\begin{list}{$\bullet$}{\leftmargin 0.25in \topsep 2pt \itemsep 1pt}
    \item $\rank_1(t)=\sum_{i=1}^{\rmk}\Prob(r(t=i))$
    \item $\rank_2(t)=\sum_{i=1}^{\rmk}\Prob(r(t=i))\cdot i$
    \item $\rank_3(t,i)=\sum_{j=1}^{\rmk}\Prob(r(t)=j))|i-j|+i\Prob(r(t)>\rmk)$.
\end{list}
\eat{
\amolnote{Regarding the last one: Does $\Prob(r(t)>\rmk)$ include the probability that $t$ is not in the possible world? In other
words, should this be $\Prob(r(t) \not\le \rmk)$ instead of $\Prob(r(t) > \rmk)$?}
}
It is easy to see $\rank_1(t),\rank_2(t),\rank_3(t)$ can be computed in polynomial time
for a probabilistic and/xor tree using our generating functions method.

A careful and non-trivial rewriting of $E_{pw\in PW}[F^*(\tau,\tau_{pw})]$ shows that it also
has the form (Figure \ref{fig:spearmans}):
\[ E_{pw\in PW}[F^*(\tau,\tau_{pw})] = C + \sum_{t\in T} \sum_{i=1}^{\rmk} \delta(t=\tau(i))f(t, i) \]
where $C$ is a constant independent of $\tau$, and $f(t,i)$ is a function of $t$ and $i$, which is polynomially computable.
Figure \ref{fig:spearmans} shows the exact derivation.

Thus, we only need to minimize the second term, which can be modeled as the assignment problem
and can be solved in polynomial time.

\subsection{Kendall's Tau Distance}

Then {\em Kendall's tau} distance (also called Kemeny distance) $\dist_K$ between two \Topk\ lists
$\tau_1$ and $\tau_2$ is
defined to be the number of unordered pairs $(t_i,t_j)$ such that
that the order of $i$ and $j$ disagree in any full rankings extended from $\tau_1$ and $\tau_2$, respectively.
It is shown that $\dist_F$ and $\dist_K$ and a few other generalizations of Spearman's footrule and Kendall's tau
metrics form a big equivalence class, i.e., they are within a constant factor of each other ~\cite{fagin:sjdm}.
Therefore, the optimal solution for $\dist_F$ implies constant approximations for all metrics in
this class (the constant for $\dist_K$ is $2$).

However, we can also easily obtain a $3/2$-approximation for $\dist_K$ by extending the $3/2$-approximation
for partial rank aggregation problem due to Ailon~\cite{conf/soda/ailon07}.
The only information used in their algorithm is the proportion of lists where $t_i$ is ranked higher than $t_j$
for all $i,j$. In our case, this corresponds to $\Prob(r(t_i)<r(t_j))$. This can be easily computed in
polynomial time using the generating functions method.

We also note that the problem of optimally computing the mean answer is NP-hard for probabilistic and/xor trees.
This follows from the fact that probabilistic and/xor trees can simulate arbitrary possible worlds, and
previous work has shown that aggregating even 4 rankings under this distance metric is NP-Hard~\cite{conf/www/rankaggregation}.

\section{Other Types of Queries}
\label{sec:other types of queries}
We briefly extend the notion of consensus answers to two other types of queries and present
some initial results.

\subsection{Aggregate Queries}
\label{sec:aggregates}

\noindent{Consider} a query of the type: \\[2pt]
\centerline{\em select groupname, count(*) from R group by groupname} \\[2pt]
Suppose there are $m$ potential groups (indexed by groupname) and $n$ independent
tuples with attribute uncertainty. The probabilistic database can be specified by the matrix
$\bfP=[p_{i,j}]_{n\times m}$ where $p_{i,j}$ is the probability that
tuple $i$ takes groupname $j$ and $\sum_{j=1}^m p_{i,j}=1$ for any $1\leq i\leq n$.
A query result (on a deterministic relation) is a $m$-dimensional vector $\mathbf{r}$ where the $i^{th}$ entry is
the number of tuples having groupname $i$.
The natural distance metric to use is the squared vector distance.

Computing the mean answer is easy in this case, because of linearity of expectation: we simply
take the mean for each aggregate separately, i.e.,
$\bfbr = \mathbf{1}\bfP$ where $\mathbf{1}=(1,1,\ldots,1)$.
We note the mean answer minimizes the expected squared vector distance to any possible answer.

The median world requires that the returned answer be a possible answer.
It is not clear how to solve this problem optimally in polynomial time.
To enumerate all worlds is obviously not computationally feasible.
Rounding entries of $\bfbr$ to the nearest integers may not result in a possible answer.

Next we present a polynomial time algorithm to find a closest possible answer
to the mean world $\bfbr$.
This yields a $4$-approximation for finding the median answer.
We can model the problem as follows:
Consider the bipartite graph $B(U,V,E)$ where each node in $U$ is a tuple, each node in $V$
is a groupname, and an edge $(u,v), u\in U, v\in V$ indicates that tuple $u$ takes groupname $v$ with non-zero
probability.
We call a subgraph $G'$ such that $deg_{G'}(u)=1$ for all $u\in U$
and $deg_{G'}(v)=\bfr[v]$,  an {\em $\bfr$-matching} of $B$ for some $m$-dimensional integral vector $\bfr$.
Given this, our objective is to find an $\bfr$-matching of $B$
such that $||\bfr-\bfbr||_2$ is minimized.
Before presenting the main algorithm, we need the following lemma.
\eat{
\begin{lemma}
\label{lm_floor}
Suppose $\bftr$ is the vector obtained by rounding down each entry of $\bfbr$ to the nearest integer.
There is a subset of $\bftr\cdot \mathbf{1}$ tuples that generates $\bftr$ with non-zero probability.
\end{lemma}
\begin{proof}
By Hall's Theorem, we only need to prove that $\sum_{v\in S}\bftr[v]\leq |N_B(S)|$ for any subset $S\in V$.
This is true simply because
$$ \sum_{v\in S}\bftr[v]=\sum_{v\in S}\sum_{u\in U}\bfP[u,v] $$
$$=\sum_{u\in N_B(S)}\sum_{v\in S}\bfP[u,v]\leq \sum_{u\in N_B(S)} 1  \leq |N_B(S)|.  $$
\qed
\end{proof}
}
\begin{lemma}
\label{lm_ceilingfloor}
The possible world $\bfr^*$ that is closest to $\bfbr$ is of the following form:
$\bfr^*[i]$ is either $\lfloor\bfbr[i]\rfloor$ or $\lceil\bfbr[i]\rceil$ for each $1\leq i\leq m$.
\end{lemma}
\begin{proof}
Let $M^*$ be the corresponding $\bfr^*$-matching.
Suppose the lemma is not true, and there exists $i$ such that
$|\bfr^*[i]-\bfbr[i]|>1$.
W.l.o.g, we assume $\bfr^*[i]>\bfbr[i]$. The other case can be proved the same way.
Consider the connected component $K=\{U',V',E(U',V')\}$ containing $i$.
We claim that there exists $j\in V'$ such that $\bfr^*[j]<\bfbr[j]$ and
there is an alternating path $P$ with respect to $M^*$
connecting $i$ and $j$.
Therefore, $M'=M^*\oplus P$ is also a valid matching. Suppose $M'$ is a $\bfr'$-matching. But:

\vspace{-15pt}
{\footnotesize
\begin{eqnarray*}
||\bfr'-\bfbr||_2^2 &=& \sum_{v=1}^m (\bfr'[v]-\bfbr[v])^2 \\
&=& \sum_{v=1}^m (\bfr^*[v]-\bfbr[v])^2-(\bfr^*[i]-\bfbr[i])^2- \\
&& (\bfr^*[j]-\bfbr[j])^2 +(\bfr'[i]-\bfbr[i])^2+(\bfr'[j]-\bfbr[j])^2 \\
&=& ||\bfr^*-\bfbr||_2^2-(\bfr^*[i]-\bfbr[i])^2-(\bfr^*[j]-\bfbr[j])^2 \\
&& +(\bfr^*[i]-1-\bfbr[i])^2+(\bfr^*[j]+1-\bfbr[j])^2 \\
&=& ||\bfr^*-\bfbr||_2^2 +2-2\bfr^*[i]+2\bfbr[i]+2\bfr^*[j]-2\bfbr[j] \\
&<& ||\bfr^*-\bfbr||_2^2.
\end{eqnarray*}
}
\vspace{-15pt}

\noindent{This} contradicts the assumption $\bfr^*$ is the vector closest to $\bfbr$.

Now, we prove the claim.
We grow a alternating path (w.r.t. $M^*$) tree rooted at $i$ in a BFS manner:
at odd depth, we extend all edges in $M^*$ and at even depth, we extend all edge not in $M^*$.
Let $O\subseteq V$ be the set of nodes at odd depth ($i$ is at depth $1$) and $E\subseteq U$
the set of nodes at even depth.
It is easy to see $N_B(E)=O$, $E\subseteq N_B(O)$ and $\sum_{v\in O}\bfr^*[v]=|E|$.
Suppose $\bfr^*[v]\geq \bfbr[v]$ for all $v$ and $\bfr^*[i]\geq \bfbr[i]$.
However, the contradiction follows since:
$$ |E|=\sum_{v\in O}\bfr^*[v]>\sum_{v\in O}\bfbr[v]=\sum_{v\in O}\sum_{u\in N_B(O)}\bfP[u,v] $$
$$ \ \ \ \ \ \ \ \ \ \ \ \ \ \ \ \ \ \ \ \ \ \ \ \ \ \ \ \ \ \ \ \ \ \ \ \ \ \ \ \ \ \ \ \ \ \ \ = \sum_{v\in O}\sum_{u\in E}\bfP[u,v] = |E|.  $$
\qed
\end{proof}

With Lemma \ref{lm_ceilingfloor} at hand, we can construct the following min-cost network flow
instance to compute the vector $\bfr^*$ closest to $\bfbr$.
Add to $B$ a source $s$ and a sink $t$.
Add edges $(s,u)$ with capacity upper bound $1$ for all $u\in U$.
For each $v\in V$ and $\bfbr[v]$ is not integer, add two edges $e_1(v,t)$ and $e_2(v,t)$.
$e_1(v,t)$ has both lower and upper bound of capacity $\lfloor \bfbr[v] \rfloor$
and $e_2(v,t)$ has capacity upper bound $1$ and cost
$(\lceil \bfbr[v] \rceil-\bfbr[v])^2-(\lfloor \bfbr[v] \rfloor-\bfbr[v])^2$.
If $\bfbr[v]$ is a integer, we only add $e_1(v,t)$.
We find a min-cost integral flow of value $n$ on this network.
For any $v$ such that $e_2(v,t)$ is saturated,
we set $\bfr^*[v]$ to be $\lceil \bfbr \rceil$ and $\lfloor \bfbr \rfloor$ otherwise.
Such a flow with minimum cost suggests the optimality of
the vector $\bfr^*$ due to Lemma \ref{lm_ceilingfloor}.

\begin{theorem}
There is a polynomial time algorithm for finding the vector $\bfr^*$ to $\bfbr$
such that $\bfr^*$ corresponds to some possible answer with non-zero probability.
\end{theorem}

\noindent{Finally}, we can prove that: 
\begin{corollary}
There is a polynomial time deterministic 4-approximation
for finding the median aggregate answer.
\end{corollary}
\begin{proof}
Suppose $\bfr^*$ is the answer closest to the mean answer $\bfbr$ and
$\bfr^m$ is the median answer. Let $\bfr$ be the vector corresponding to the random answer. Then:

\vspace{-12pt}
{\small
\begin{eqnarray*}
\Exp[\dist(\bfr^*,\bfr)]
&\leq& \Exp[2(\dist(\bfr^*,\bfbr)+\dist(\bfbr,\bfr))]
=    2\left(\dist(\bfr^*,\bfbr)+\Exp[\dist(\bfbr,\bfr)]\right) \\
&\leq& 4\Exp[\dist(\bfbr,\bfr )] \leq 4\Exp[\dist(\bfr^m,\bfr )].
\end{eqnarray*}
}
\end{proof}

\subsection{Clustering}
\label{sec:consensus}

\ignore{
Conceptually the definition would be that: find the clusterheads for each possible world, and then
take their ``mean''. This appears to be quite complicated though. Given two k-sets of tuples (each
corresponding to the answer in a different possible world), the natural distance function is the
weight of the minimum bipartite matching between these two k-sets, and it is not clear how to
incorporate something like that into the answer.

\amolnote{It appears that some form of weighted distance computation is already happening in our
definition of clustering in the paper. But it is not clear to me what its semantics are.}
}
The \concluster\ problem is defined as follows: given $k$ clusterings $\calC_1,\ldots,\calC_k$ of $V$, find a
clustering $\calC$ that minimizes $\sum_{i=1}^k \dist(\calC,\calC_i)$.
In the setting of probabilistic databases,
the given clusterings are the clusterings in the possible worlds, weighted
by the existence probability. The main problem with extending the notion
of consensus answers to clustering is that the input clusterings are not well-defined
(unlike ranking where the score function defines the ranking in any world).
We consider a somewhat simplified version of the problem, where we assume that
two tuples $t_i$ and $t_j$ are clustered together in a possible world, if and only if they take
the same value for the value attribute $A$ (which is uncertain).
Thus, a possible world $pw$ uniquely determines a clustering $\calC_{pw}$.
We define the distance between two clustering $\calC_1$ and $\calC_2$
to be the number of unordered pairs of tuples that are clustered together in $\calC_1$,
but separated in the other (the \concluster\ metric).
To deal with nonexistent keys in a possible world, we artifically create a cluster containing
all of those.

Our task is to find a mean clustering $\calC$ such that $\Exp[\dist(\calC,\calC_{pw})]$.
Approximation with factor of $4/3$ is known for \concluster~\cite{journal/jacm/ailon08}, and can be adapted to our problem
in a straightforward manner.
In fact, that approximation algorithm simply needs $w_{t_i,t_j}$ for all $t_i,t_j$, where
$w_{t_i,t_j}$ is the fraction of input clusters that cluster $t_i$ and $t_j$ together, and can
be computed as:
$ w_{t_i,t_j} ={\sum_{a\in A}\Prob(i.A=a\wedge j.A=a)}$.

\smallskip
To compute these quantities given an and/xor tree, we
associate a variable $x$ with all leaves with value $(i,a)$ and $(j,a)$, and constant $1$ with the other leaves.
From Theorem \ref{thm_generating}, $\Prob(i.A=a\wedge j.A=a)$ is simply the coefficient of $x^2$
in the corresponding generating function.

\newpage
\section{Conclusion}
We addressed the problem of finding a single representative answer to a query over probabilistic
databases by generalizing the notion of inconsistent information integration. We believe this
approach provides a systematic and formal way to reason about the semantics of probabilistic query
answers, especially for \Topk\ queries. Our initial work has opened up many interesting
avenues for future work. These include design of efficient exact and approximate algorithms for finding consensus
answers for other types of queries,
exploring connections to safe plans, and understanding the semantics of
the other previously proposed ranking functions using this framework.

{\footnotesize
\bibliographystyle{plain}
\bibliography{probdb}

\begin{thebibliography}{10}

\bibitem{conf/soda/ailon07}
Nir Ailon.
\newblock Aggregation of partial rankings, p-ratings and top-m lists.
\newblock In {\em SODA}, pages 415--424, 2007.

\bibitem{journal/jacm/ailon08}
Nir Ailon, Moses Charikar, and Alantha Newman.
\newblock Aggregating inconsistent information: Ranking and clustering.
\newblock In {\em J.ACM}, volume 55(5), 2008.

\bibitem{andritsos:icde06}
Periklis Andritsos, Ariel Fuxman, and Renee~J. Miller.
\newblock Clean answers over dirty databases.
\newblock In {\em ICDE}, 2006.

\bibitem{antova:sigmod07}
Lyublena Antova, Christoph Koch, and Dan Olteanu.
\newblock From complete to incomplete information and back.
\newblock In {\em SIGMOD}, 2007.

\bibitem{barbara:kde92}
B., H.~Garcia-Molina, and D.~Porter.
\newblock The management of probabilistic data.
\newblock {\em IEEE TKDE}, 1992.

\bibitem{beskales:vldb08}
George Beskales, Mohamed~A. Soliman, and Ihab~F. Ilyas.
\newblock Efficient search for the top-k probable nearest neighbors in
  uncertain databases.
\newblock In {\em VLDB}, 2008.

\bibitem{buckles}
B.~Buckles and F.~E. Petry.
\newblock A fuzzy model for relational databases.
\newblock {\em Intl. Journal of Fuzzy Sets and Syst.}, 1982.

\bibitem{cheng:sigmod03}
Reynold Cheng, Dmitri Kalashnikov, and Sunil Prabhakar.
\newblock {E}valuating probabilistic queries over imprecise data.
\newblock In {\em SIGMOD}, 2003.

\bibitem{Cormode09}
Graham Cormode, Feifei Li, and Ke~Yi.
\newblock Semantics of ranking queries for probabilistic data and expected
  ranks.
\newblock In {\em ICDE}, 2009.

\bibitem{cormode:pods08}
Graham Cormode and Andrew McGregor.
\newblock Approximation algorithms for clustering uncertain data.
\newblock In {\em PODS}, 2008.

\bibitem{dalvi:vldb04}
Nilesh Dalvi and Dan Suciu.
\newblock Efficient query evaluation on probabilistic databases.
\newblock In {\em VLDB}, 2004.

\bibitem{conf/pods/DalviS07}
Nilesh Dalvi and Dan Suciu.
\newblock Management of probabilistic data: Foundations and challenges.
\newblock In {\em PODS}, 2007.

\bibitem{deshpande:vldb04}
Amol Deshpande, Carlos Guestrin, Sam Madden, Joseph~M. Hellerstein, and Wei
  Hong.
\newblock Model-driven data acquisition in sensor networks.
\newblock In {\em VLDB}, 2004.

\bibitem{conf/www/rankaggregation}
C.~Dwork, R.~Kumar, M.~Naor, and D.~Sivakumar.
\newblock Rank aggregation methods for the web.
\newblock In {\em Proceedings of the Tenth International Conference on the
  World Wide Web (WWW)}, pages 613--622, 2001.

\bibitem{dwork_rankaggr_revisited}
C.~Dwork, R.~Kumar, M.~Naor, and D.~Sivakumar.
\newblock Rank aggregation revistied.
\newblock In {\em Manuscript}, 2001.

\bibitem{fagin:sjdm}
Ronald Fagin, Ravi Kumar, and D.~Sivakumar.
\newblock Comparing top k lists.
\newblock {\em SIAM J. Discrete Mathematics}, 17(1):134--160, 2003.

\bibitem{fuhr:is97}
N.~Fuhr and T.~Rolleke.
\newblock A probabilistic relational algebra for the integration of information
  retrieval and database systems.
\newblock {\em ACM Trans. on Info. Syst.}, 1997.

\bibitem{debulletin:march2006}
Minos Garofalakis and Dan Suciu, editors.
\newblock {\em IEEE Data Engineering Bulletin Special Issue on Probabilistic
  Data Management}.
\newblock March 2006.

\bibitem{grahne}
Gosta Grahne.
\newblock Horn tables - an efficient tool for handling incomplete information
  in databases.
\newblock In {\em PODS}, 1989.

\bibitem{gupta:vldb06}
Rahul Gupta and Sunita Sarawagi.
\newblock Creating probabilistic databases from information extraction models.
\newblock In {\em VLDB}, Seoul, Korea, 2006.

\bibitem{conf/icde/HuaPZL08}
M.~Hua, J.~Pei, W.~Zhang, and X.~Lin.
\newblock Efficiently answering probabilistic threshold top-k queries on
  uncertain data.
\newblock In {\em ICDE}, 2008.

\bibitem{conf/sigmod/HuaPZL08}
M.~Hua, J.~Pei, W.~Zhang, and X.~Lin.
\newblock Ranking queries on uncertain data: A probabilistic threshold
  approach.
\newblock In {\em SIGMOD}, 2008.

\bibitem{imielinski:jacm84}
T.~Imielinski and W.~{Lipski, Jr.}
\newblock Incomplete information in relational databases.
\newblock {\em Journal of the ACM}, 1984.

\bibitem{DBLP:conf/pods/JayramMMV07}
T.~S. Jayram, Andrew McGregor, S.~Muthukrishnan, and Erik Vee.
\newblock Estimating statistical aggregates on probabilistic data streams.
\newblock In {\em PODS}, pages 243--252, 2007.

\bibitem{Borda81}
J.C.Borda.
\newblock M\'{e}moire sur les \'{e}lections au scrutin.
\newblock {\em Histoire de l'Acad\'{emie} Royale des Sciences}, 1781.

\bibitem{Kemeny59}
J.G.Kemeny.
\newblock Mathematics without numbers.
\newblock {\em Daedalus}, 88:571--591, 1959.

\bibitem{book/hodge00}
J.Hodge and R.E.Klima.
\newblock {\em The mathematics of voting and elections: a hands-on approach}.
\newblock AMS, 2000.

\bibitem{lakshmanan:tods97}
L.~Lakshmanan, N.~Leone, R.~Ross, and V.~S. Subrahmanian.
\newblock Probview: a flexible probabilistic database system.
\newblock {\em ACM Trans. on DB Syst.}, 1997.

\bibitem{tech}
Jian Li, Barna Saha, and Amol Deshpande.
\newblock Ranking and clustering in probabilistic databases.
\newblock \url{http://www.cs.umd.edu/~lijian/paper/clusterrank_tr.pdf}, 2008.
\newblock Unpublished manuscript.

\bibitem{conf/focs/matching80}
Silvio Micali and Vijay~V. Vazirani.
\newblock An $o(sqrt(|v|) |e|)$ algorithm for finding maximum matching in
  general graphs.
\newblock In {\em FOCS '80: Proceedings of the 21th Annual Symposium on
  Foundations of Computer Science}, pages 17--27, 1980.

\bibitem{book/condorcet}
M.J.Condorcet.
\newblock {\em \'{E}ssai sur l'application de l'analyse \`{a} la
  probabilit\'{e} des d\'{e}cisions rendues \`{a} la pluralit\'{e} des voix}.
\newblock 1785.

\bibitem{re:icde07}
Christopher Re, Nilesh Dalvi, and Dan Suciu.
\newblock Efficient top-k query evaluation on probabilistic data.
\newblock In {\em ICDE}, 2007.

\bibitem{re:vldb07}
Christopher Re and Dan Suciu.
\newblock Materialized views in probabilistic databases for information
  exchange and query optimization.
\newblock In {\em VLDB}, Vienna, Austria, 2007.

\bibitem{sarma:icde06}
A.~Sarma, O.~Benjelloun, A.~Halevy, and J.~Widom.
\newblock Working models for uncertain data.
\newblock In {\em ICDE}, 2006.

\bibitem{sen:icde07}
Prithviraj Sen and Amol Deshpande.
\newblock Representing and querying correlated tuples in probabilistic
  databases.
\newblock In {\em ICDE}, 2007.

\bibitem{sen:vldb08}
Prithviraj Sen, Amol Deshpande, and Lise Getoor.
\newblock Exploiting shared correlations in probabilistic databases.
\newblock In {\em VLDB}, 2008.

\bibitem{soliman:icde07}
M.~Soliman, I.~Ilyas, and K.~C. Chang.
\newblock Top-k query processing in uncertain databases.
\newblock In {\em ICDE}, 2007.

\bibitem{conf/dbpl/re07}
Christopher R\'{e}and~Dan Suciu.
\newblock Efficient evaluation of having queries on a probabilistic database.
\newblock In {\em DBPL}, 2007.

\bibitem{wang:vldb08}
Daisy~Zhe Wang, Eirinaios Michelakis, Minos Garofalakis, and Joseph~M.
  Hellerstein.
\newblock {BayesStore}: Managing large, uncertain data repositories with
  probabilistic graphical models.
\newblock In {\em VLDB}, Auckland, New Zealand, 2008.

\bibitem{widom:cidr05}
J.~Widom.
\newblock Trio: A system for integrated management of data, accuracy, and
  lineage.
\newblock In {\em CIDR}, 2005.

\bibitem{conf/icde/YiLSK08}
Ke~Yi, Feifei Li, Divesh Srivastava, and George Kollios.
\newblock Efficient processing of top-k queries in uncertain databases.
\newblock In {\em ICDE}, 2008.

\bibitem{Wakabayashi98}
Y.Wakabayashi.
\newblock The complexity of computing medians of relations.
\newblock In {\em Resenhas}, volume 3(3), pages 323--349, 1998.

\bibitem{conf/dbrank/ZhangC08}
Xi~Zhang and Jan Chomicki.
\newblock On the semantics and evaluation of top-k queries in probabilistic
  databases.
\newblock In {\em DBRank}, 2008.

\end{thebibliography}
}
\end{document}